\documentclass[a4paper,USenglish,cleveref, autoref, thm-restate]{lipics-v2021}

\pdfoutput=1 %
\hideLIPIcs  %

\nolinenumbers

\usepackage{csquotes}
\usepackage{amsmath}
\usepackage{mathtools}
\usepackage{multirow}
\usepackage{nicematrix}
\usepackage{todonotes}
\usepackage{mdframed}
\usepackage{bm}
\usepackage{todonotes}
\usepackage{cite}
\usepackage{color, colortbl} %
\usepackage{xspace}
\usepackage{nicefrac}
\usepackage{booktabs} %
\usepackage{relsize}
\usepackage{siunitx}
\usepackage{bbold}
\usepackage{pgfplots}
\usepackage{enumitem}
\usepackage{placeins}

\setlist{topsep=0mm}
\Crefformat{figure}{#2Fig.~#1#3}
\Crefmultiformat{figure}{Figs.~#2#1#3}{ and~#2#1#3}{, #2#1#3}{ and~#2#1#3}

\definecolor{lightGray}{gray}{0.9}

\newmdenv[
  topline=false,
  bottomline=false,
  rightline=false,
  skipabove=4pt,
  skipbelow=1pt,
  innerleftmargin=3pt,
  innerrightmargin=0pt,
  innertopmargin=1pt,
  innerbottommargin=1pt,
  linecolor=orange
]{siderules}

\newcommand{\Rho}{R}

\newcommand{\PP}{\mathcal{P}}

\newcommand{\myparagraph}[1]{\medskip\noindent\textsf{\textbf{#1}}}
\newcommand{\mysubparagraph}[1]{\smallskip\noindent\textsf{\textbf{#1}}}

\newcommand{\flowvar}[1]{\ensuremath{f^{uv}_{(#1)}}\xspace}
\newcommand{\flowvarij}{\flowvar{i,j}}
\newcommand{\flowvarji}{\flowvar{j,i}}

\newcommand{\ABplus}{AB\textsuperscript{+}}
\newcommand{\kSP}[1][k]{\boldmath$#1$-SP}
\newcommand{\BiAstar}[1]{$#1$-BiA$^*$}

\newcommand{\weightSet}{\mathcal{W}}
\newcommand{\algo}[2][]{\textsf{#2}\textsubscript{\textbf{\textsf{#1}}}\xspace}
\newcommand{\instance}[1]{\textsf{#1}\xspace}
\newcommand{\problem}[1]{\textsf{#1}\xspace}
\newcommand{\weight}[1]{\ensuremath{\mathrm{W_{#1}}}\xspace}

\newcommand{\Hyp}[1]{Hypothesis \hypShort{#1}}
\newcommand{\hyp}[1]{hypothesis \hypShort{#1}}
\newcommand{\hypShort}[1]{$\mathcal{H}#1$\xspace}
\newcommand{\tableHead}{
    \hline
    weight &  \multicolumn{6}{c|}{\weight{n}} &  \multicolumn{6}{c|}{\weight{euc}} &  \multicolumn{4}{c|}{\weight{1} (=unweighted)}\\
    \hline
    density &  \multicolumn{3}{c|}{any $\delta\in D$} &  \multicolumn{3}{c|}{any $\varrho\in \Rho$} &  \multicolumn{3}{c|}{any $\delta\in D$} &  \multicolumn{3}{c|}{any $\varrho\in \Rho$} & \multicolumn{2}{c|}{any $\delta\in D$} &  \multicolumn{2}{c|}{any $\varrho\in \Rho$}  \\
    \hline
    $\alpha$ & $1.2$ & $2$ & $3$ & $1.2$ & $2$ & $3$ & $1.2$ & $2$ & $3$ & $1.2$ & $2$ & $3$ & $2$ & $3$ & $2$ & $3$ \\
}

\makeatletter
\newenvironment{taggedsubequations}[1]
 {\addtocounter{equation}{-1}%
  \begin{subequations}%
  \def\@currentlabel{#1}%
  \renewcommand{\theequation}{#1.\alph{equation}}%
 }
 {\end{subequations}}
\makeatother

\title{Exact Minimum Weight Spanners via Column Generation} 

\titlerunning{Exact \problem{MWSP} via Column Generation} 

\author{Fritz Bökler}
{Institute of Computer Science, Osnabrück University, Germany}
{fboekler@uos.de}
{https://orcid.org/0000-0002-7950-6965}
{}

\author{Markus Chimani}
{Institute of Computer Science, Osnabrück University, Germany}
{markus.chimani@uos.de}
{https://orcid.org/0000-0002-4681-5550}
{}

\author{Henning {Jasper}\footnote{Corresponding author}}
{Institute of Computer Science, Osnabrück University, Germany}
{henning.jasper@uos.de}
{https://orcid.org/0000-0002-9821-8600}
{}%

\author{Mirko H. {Wagner}}
{Institute of Computer Science, Osnabrück University, Germany}
{mirko.wagner@uos.de}
{https://orcid.org/0000-0003-4593-8740}
{}

\authorrunning{F. Bökler, M. Chimani, H. Jasper and M. H. Wagner}  %

\Copyright{Fritz Bökler, Markus Chimani, Henning Jasper and Mirko H. Wagner} %



\category{} %

\relatedversion{} %

\EventEditors{--Editor Names--}
\EventNoEds{2}
\EventLongTitle{Submitted to ESA 2024}
\EventShortTitle{Subm.~to~ESA~24}
\EventAcronym{---}
\EventYear{2016}
\EventDate{Month Day--Day, 2024}
\EventLocation{City, Country}
\EventLogo{}
\SeriesVolume{X}
\ArticleNo{X}

\begin{document}

\maketitle

\begin{abstract}
Given a weighted graph $G$, a minimum weight $\alpha$-spanner is a least-weight subgraph $H\subseteq G$ that preserves minimum distances between all node pairs up to a factor of $\alpha$. There are many results on heuristics and approximation algorithms, including a recent investigation of their practical performance~\cite{chimani2022}. Exact approaches, in contrast, have long been denounced as impractical:
The first exact ILP (integer linear program) method~\cite{Sigurd2004} from 2004 is based on a model with exponentially many path variables, solved via column generation. A second approach~\cite{ahmed2019}, modeling via arc-based multicommodity flow, was presented in 2019. In both cases, only graphs with 40--100 nodes were reported to be solvable.

In this paper, we briefly report on a theoretical comparison between these two models from a polyhedral point of view, and then concentrate on improvements and engineering aspects. We evaluate their performance in a large-scale empirical study. We report that our tuned column generation approach, based on
multicriteria shortest path computations, is able to solve instances with over 16\,000 nodes within 13\,min.
Furthermore, now knowing optimal solutions for larger graphs, we are able to investigate the quality of the strongest known heuristic on reasonably sized instances for the first time.
\end{abstract}

\newpage
\setcounter{page}{1}
\section{Introduction} \label{sec:intro}
Let $G=(V,E)$ be an undirected graph with $n$ nodes and $m$ edges.
The distance $d_G(u, v)$ is the length of a shortest path between two nodes $u$ and $v$ in $G$, possibly subject to positive edge weights $w_e > 0$ for all $e \in E$.
A \emph{spanner} is a subgraph $H\subseteq G$ that preserves these distances within some quality degree; see Ahmed et.\ al \cite{ahmed2020} for an overview on several variants.
In this paper, we consider the original and most prominent variant of a \emph{multiplicative $\alpha$-spanner} (in the following just \emph{($\alpha$-)spanner} for simplicity):
for a given \emph{stretch factor} $\alpha\geq 1$, we require that the \emph{stretch constraint}
$d_H(u,v) \leq \alpha \cdot d_G(u,v)$ holds for all node pairs $\{u,v\} \in \binom{V}{2}$.
The \emph{minimum weight spanner problem (\problem{MWSP})} is thus to find such a spanner of minimum total weight.
For uniform edge weights, i.e., $w_e =1$ for all $e\in E$, \problem{MWSP} is equivalent to finding a spanner of minimum size $|E(H)|$.
Spanners were first introduced %
in the context of synchronization in distributed systems and communication networks \cite{peleg1989, Ullman1989}.
Their efficient computation is a highly relevant topic in many applications, e.g., routing problems~\cite{shpungin2010}, graph drawings~\cite{wallinger2023}, access control hierarchies~\cite{jha2013, bhattacharyya2012}, or passenger assignment~\cite{heinrich2023}.

\problem{MWSP} is known to be NP-hard~\cite{CAI1994187}.
Thus, most published algorithms are heuristics or approximations.
However, their guarantees often primarily approximate, e.g., the ratio between the weight of the spanner and that of a minimum spanning tree (the so-called \emph{lightness})---see~\cite{ahmed2020,chimani2022} for an overview.
One of the earliest \problem{MWSP}-algorithms
is the \emph{Basic Greedy} (\algo{BG}) algorithm by Althöfer et al.~\cite{althofer1993}.
Several algorithms were developed in an attempt to improve over \algo{BG}~\cite{roditty2011,baswana2007,berman2013,elkin2017}; some of them allow the spanners to violate the stretch constraint by a factor of $1 + \varepsilon$~\cite{chandra1992, elkin2015, elkin2016, alstrup2022}.
However, \algo{BG} still is beneficial w.r.t.\ most guarantees and has been proven to be \emph{existentially optimal}~\cite{ahmed2020}; also, most of the newer algorithms lead to very complex, non-practical implementations.
Recently,~\cite{chimani2022} investigated the practical performance of the most promising of those approaches. They conclude that in almost all cases, \algo{BG} provides the sparsest and lightest spanners, typically even within the shortest running time.

There are exact algorithms for some special cases of \problem{MWSP}.
Cai and Keil~\cite{CaiKeil1994} present a linear time algorithm for minimum $2$-spanners in unweighted graphs with maximum degree at most four. 
Kobayashi~\cite{kobayashi2018} gives an FPT algorithm for unweighted \problem{MWSP} that is parameterized in the number of edges that need to be removed to yield $H$. 

For general \problem{MWSP}, however, there are currently only two published exact algorithms, both solving an integer linear program (ILP):
The first algorithm was proposed in 2004 by Sigurd and Zachariasen~\cite{Sigurd2004} and requires column generation:
They use a \emph{path-based} ILP formulation containing an exponential number of path variables, which are incrementally introduced by solving the \emph{pricing problem}: a particular kind of the \emph{(Resource) Constrained Shortest Path} problem (\problem{CSP})~\cite{pugliese2013,festa2015}.
Even in the case of only two resources (such that one is constrained, while the other is minimized), \problem{CSP} is NP-hard~\cite{Garey1979}, but can often be solved effectively~\cite{pugliese2013,festa2015}.
\problem{CSP} is regularly used as a building block within column generation, e.g., in vehicle routing \cite{baldacci2012, zhu2012}, aircraft flight assignment \cite{barnhart1998}, and crew scheduling \cite{mingozzi1999}. 
The approach of~\cite{Sigurd2004} was tested on graphs with up to 64 nodes, but not every instance could be solved within a time limit of 30 minutes.

The second exact approach was proposed by Ahmed et al.~\cite{ahmed2019} in 2019. Their model uses an \emph{arc-based} multicommodity flow formulation
and has polynomial size.
While not directly comparing their approach to~\cite{Sigurd2004}, they tested their formulation on graphs with up to 100 nodes, %
on which their solver needed up to 40 hours.

\myparagraph{Contribution.} We compare the known exact ILP approaches for the spanner problem for the first time, and improve on them. Our goal is to show that, despite the results suggested in literature, exact approaches for the spanner problem are in fact a worthwhile endeavor in practice.
On the theory side, we investigate their relative polyhedral strength. From the practical point of view, the arc-based approach is relatively straight-forward to implement, but the path-based approach turns out to be much more interesting and fruitful w.r.t.\ boosting its performance: we propose several improvements by means of size reduction, new initialization strategies, and stronger pricing algorithms, facilitating concepts from multiobjective optimization.
Our modifications allow us to solve instances orders of magnitudes larger than before, e.g., road networks with over 16\,000 nodes within 13 minutes. Our path-based column generation approach is significantly faster and often even yields smaller models than the polynomially-sized (and further tuned) arc-based model.
Lastly, our results allow us to further evaluate the quality of the strongest known spanner heuristic \algo{BG}~\cite{althofer1993}.
In contrast to previous works, we can now investigate the quality on reasonably-sized instances w.r.t.\ optimal objective values.

\section{Original Column Generation Approach} \label{sec:ColGen}
We first summarize the column generation approach for \problem{MWSP}~\cite{Sigurd2004}, and discuss our improvements later in \Cref{sec:AE}. The set of \emph{terminal pairs} $K$ contains all node pairs for which the stretch constraint is enforced.
In~\cite{Sigurd2004}, they use $K=V \times V$ consisting of all \emph{ordered} node pairs. Herein, we prefer the sufficient $K=\binom{V}{2}$ of unordered pairs, to avoid redundancy. 

The key concept of the model is to establish a binary \emph{path variable} $y_P$ for each
path $P$ in $G$, which is $1$ if and only if $P$ is contained in the solution $H$ and at the same time is used to witness that its endpoints $u,v$ satisfy the stretch constraint $d_H(u,v) \leq \alpha \cdot d_G(u,v)$. For all $\{u,v\}\in K$, let $\PP_{uv}$ denote the set of all $u$-$v$-paths that are no longer than $\alpha \cdot d_G(u,v)$, and let $\PP= \bigcup_{\{u, v\} \in K} \PP_{uv}$. 
We can write the ILP model \eqref{ilp:PB}, where decision variables $x_e$ establish the solution $H$:
\begin{taggedsubequations}{PB}\label{ilp:PB}
\begin{align}
\mathrlap{\text{(PB)}}&&
   \min  \sum_{e \in E} w_{e} x_{e}  &  &  &  \\
   &&\text{s.t.\ }\sum_{P \in \PP_{uv}}  y_P   & \geq 1       &  & \forall \{u, v\} \in K                  \label{PB:path}\\
   &&\sum_{P \in \PP_{uv} : e\in P} y_P & \leq x_e     &  & \forall e \in E,\, \forall \{u, v\} \in K \label{PB:edge}\\
   &&x_{e}, y_{P}                      & \in  \{0,1\} &  & \forall e \in E,\, \forall P \in \PP
\end{align}%
\end{taggedsubequations}%

To solve this ILP via branch-and-bound (B\&B), we need to solve its LP relaxation (i.e., the binary requirements are relaxed to the interval $[0,1]$) at each B\&B node.
While there are only $(|E|+1)|K|$ constraints, the exponential number of 
path variables constitutes the main challenge, which is solved using \emph{column generation}---see, e.g., \cite{conforti2014, wolsey2020} for introductions into exact branch-and-price algorithms.
We consider the \emph{restricted master problem} (RMP), which considers all edge variables but only a subset $\PP'=\bigcup_{\{u, v\} \in K} \PP'_{uv}$ of the path variables, where
$\emptyset\neq \PP'_{uv} \subseteq \PP_{uv}$. %
The original publication~\cite{Sigurd2004} does not describe the initialization of the individual sets $P'_{uv}$.
However, since singletons suffice and there is no heuristic used to yield upper bounds (and corresponding candidate paths), we assume the natural initialization that $P'_{uv}$ consists of a single shortest $u$-$v$-path in $G$, for each $\{u,v\} \in K$. We denote this initialization strategy \algo{\kSP[1]} in the following.

During the solving process, paths $P \in \PP\setminus \PP'$ are iteratively added to $\PP'$, extending the RMP.
This procedure is known as column generation.
Consider the dual LP of the relaxation of the full primal model~\eqref{ilp:PB}, where $\sigma_{uv}$ and $\pi^{uv}_e$ are the non-negative dual variables corresponding to the primal constraints \eqref{PB:path} and \eqref{PB:edge} (after canonicalization), respectively. We are mainly interested in the dual constraints
\begin{align}
\sum_{e \in P} \pi^{uv}_e & \leq \sigma_{uv} && \forall P \in \PP_{uv}, \forall \{u, v\} \in K.
\tag{PB-D}\label{eq:dual}
\end{align}
They are the only constraints that could be violated if we do not use a sufficiently large subset $\PP'$ when finding a dual solution to the RMP.
Consider a pair of optimal primal and dual solutions for the RMP. If the dual solution satisfies all constraints \eqref{eq:dual}, then by the weak duality theorem, the primal solution is optimal for the full LP relaxation of the ILP model~(\ref{ilp:PB}), instead of only for the RMP, and we can stop the column generation routine.

Otherwise, the dual solution violates at least one constraint \eqref{eq:dual} for some path $P \in \PP_{uv} \setminus \PP'$ and $\{u,v\} \in K$.
Our goal is thus to find such a $P$ with negative \emph{reduced cost} $r^\pi_\sigma(P)\coloneqq \sigma_{uv} - \sum_{e \in P} \pi^{uv}_e$, and add its path variable $y_P$ to $\PP'$.
This is called the \emph{pricing problem}. Typically, we ask for a path $P$ with smallest possible reduced cost.
We solve the pricing problem separately for every $\{u,v\} \in K$. We ask for a shortest $u$-$v$-path $P$ in $G$ w.r.t.\ edge costs $\pi^{uv}_e$, while ensuring that $P\in \PP_{uv}$.
Thus, we have a \problem{CSP} (constrained shortest path) problem, since the latter property requires us to also consider the original edge weights $w_e$ (distinct from the edge costs $\pi^{uv}_e$) and requiring that the identified path has a total weight of at most $B_{uv}=\alpha \cdot d_G(u,v)$.
Also, we are only interested in paths of total cost strictly less than $\sigma_{uv}$, as those yield negative reduced costs.

\mysubparagraph{\algo{BasicCSP}.} In \cite{Sigurd2004}, a very basic \algo{CSP} algorithm similar to~\cite{desrochers1988} is proposed, with the addition of early discarding of infeasible paths: A \emph{label} is a tuple storing the cost and weight of a path. At each node of $G$, we store a list of labels (initially $\emptyset$) sorted by cost. We start with label $(0,0)$ at $u$, and always proceed with the lowest-cost label: we propagate the label to adjacent nodes (increasing the cost and weight according to the cost and weight of the traversed edge).
We discard labels that are dominated (i.e., there is another label element-wise smaller or equal) or guaranteed to exceed the cost or weight limit. This unidirectional search stops once $v$ is reached, yielding a single path $P\in \PP_{uv}$ to add to $\PP'$.
We call this algorithm \algo{BasicCSP}.

\mysubparagraph{Minor B\&B considerations.}
To speed up the computation,~\cite{Sigurd2004} only 
adds constraints \eqref{PB:edge} to the RMP once at least one corresponding path variable exists in the RMP.
Further, whenever an edge variable $x_e$ is fixed to $0$ during branching, no further path containing $e$ needs to be considered in this subproblem. By locally setting 
the cost of $e$ to $\sigma_{uv}$, such paths will be automatically pruned by the \problem{CSP} computation for $\{u,v\}\in K$.

\section{Speedup Techniques and Algorithm Engineering} \label{sec:AE}
We are now ready to discuss our techniques to speed up the above approach. %

\mysubparagraph{General: Graph representation.}
All our shortest path implementations (\algo{\kSP[1]} and \algo{BasicCSP} above, and \algo{\kSP} and \algo{\BiAstar{\mu}} below)
use an array-based forward-star
representation~\cite{ebert1987} of the bidirected graph\footnote{Directed graphs are represented by an array of subarrays (one subarray per node); each subarray stores the outgoing edges for the respective node. Undirected graphs are encoded as bidirected graphs.}, which is especially suited for cache-efficient shortest path computations. %

\mysubparagraph{Preprocessing: Metrication.}
A minimum weight spanner will never include an edge $\{i,j\} \in E$ with $w_{\{i,j\}}>d_G(i,j)$. %
Thus, we can safely remove all such edges from $G$.

\mysubparagraph{Model size and number of \problem{CSP} calls: Terminal pairs.}
To observe the stretch constraint for every node pair in an undirected graph, it is natural to set $K=\binom{V}{2}$. 
However, it is long known that ensuring the stretch constraint for all \emph{adjacent} node pairs already guarantees feasible $\alpha$-spanners~\cite{peleg1989}. Thus, $K=E$ suffices. 
For sparse graphs, this decreases $K$'s size by a linear factor and, most importantly, also speeds up each pricing step, since we need to run a \problem{CSP} algorithm for each node pair in $K$.

\mysubparagraph{Initialization: Variable sets.}
Let $H'$ be an $\alpha$-spanner computed by a heuristic.
For each $\{u,v\} \in K$, we can initialize $P'_{uv}$ by a shortest $u$-$v$-path in $H'$.
Clearly, this provides a feasible initial set of paths.
Also, $H'$ yields an upper bound for the B\&B computation.
Comparing several theoretically strong algorithms, \cite{chimani2022} identify \algo{BG}~\cite{althofer1993} as the by far practically strongest approach. The runner-ups Baswana and Sen~\cite{baswana2007} and Berman et al.~\cite{berman2013} are only worthwhile in certain scenarios. 
Using the implementations of~\cite{chimani2022}, pilot studies suggest that the latter two are clearly inferior to \algo{BG} w.r.t.\ the initialization of the RMP.

Moreover, we generalize the \algo{\kSP[1]} initialization (used in~\cite{Sigurd2004}) such that multiple $u$-$v$-paths can be included:
For all $\{u,v\} \in K$, our \algo{\kSP[k]} initialization shall compute the $k$-shortest $u$-$v$-paths that are no longer than $\alpha \cdot d_G(u,v)$.
We implement a $k$-shortest path A$^{*}$ algorithm similar to~\cite{liu2001}.
The goal-directing distance heuristics are also used to discard paths early that are guaranteed to violate the stretch constraints.
Finding paths with \algo{\kSP[k]} is more efficient than with \problem{CSP} algorithms, but we cannot be certain that added paths locally improve the solution value of the RMP.
Note that \algo{\kSP[k]} can be modified to provide \emph{all} $u$-$v$-paths that are no longer than $\alpha \cdot d_G(u,v)$, i.e., the entire set $P_{uv}$. We call this the \emph{brute force} initialization.
We compared combinations of \algo{BG} with \algo{\kSP[k]} for $k \in \{5, 10, 20, 50\}$, and brute force to compute initial paths. Pilot studies show that \algo{\kSP[10]+BG} is the best allrounder.

\mysubparagraph{Initialization: Fixing variables.}
If there is only a single $u$-$v$-path $P$ that is no longer than $\alpha \cdot d_G(u,v)$, for some $\{u,v\} \in K$, we can fix its corresponding  path variable $y_P$ (and associated edge variables) to $1$.
Such paths are detected if a \algo{\kSP[2]} algorithm yields only a single path.
Thus, solvers using a \algo{\kSP[k]} initialization with $k>1$, can fix these variables with no overhead.

\mysubparagraph{Pricing: \algo{\BiAstar{\bm{\mu}}}.}
In the pricing step, we ask for a feasible (in terms of the stretch constraint) $u$-$v$-path $P$ with negative reduced cost, %
for each $\{u,v\}\in K$.
Formally, any such path suffices to locally improve the current solution; a path with minimum cost yields the steepest descent w.r.t.\ the objective function.
However, thinking about the pivot operation in the simplex algorithm, we know that another path with slightly higher cost may yield an overall better solution as it may allow the corresponding primal variable to be set to a higher value. Thus, we are interested in  generalizing the pricing problem from \problem{CSP} to what we call \problem{$\mu$-CSP}, for $1 \leq \mu \in \mathbb{N}$: the goal is to find a feasible path for each of the $\mu$ smallest among the negative reduced cost values, as long as they exist.
The standard \problem{CSP} is thus identical to \problem{$1$-CSP}.

For \problem{$1$-CSP}, we can use the bidirectional A$^{*}$ (\algo{BiA$^{*}$}) label-setting algorithm, as presented by Thomas et al.~\cite{thomas2019}, instead of \algo{BasicCSP}. We chiefly summarize its main ideas:
We simultaneously conduct a forward search from $u$, and a backward search from $v\in V$. By computing traditional single-source shortest paths (both from $u$ and from $v$) individually w.r.t.\ only costs or only weights, we obtain feasible lower bounds for early termination (similar to \algo{BasicCSP}) and the goal direction.
Label dominations can then be considered w.r.t.\ the induced lower bound estimations of the full $u$-$v$-paths.
Unprocessed non-dominated labels are held in a min-heap, using a lexicographic comparison on this estimated cost and weight (in this order).
Whenever a label is processed, we additionally try to join it with a suitable label from the other search direction at the respective node.
The algorithm has identified a minimum-cost feasible path (if it exists) once a joined label is processed for the first time.

We generalize this \problem{$1$-CSP} \algo{BiA$^{*}$} algorithm to our scenario for general \problem{$\mu$-CSP}, and call this variant \algo{\BiAstar{\mu}}. Thereby, we also integrate ideas from multiobjective optimization, in particular, the lexicographic-order-based label-setting algorithm of~\cite{martins1984}. The core insight is that by allowing the algorithm to proceed after its standard termination criterion, we obtain a sequence of feasible pair-wise non-dominating $u$-$v$-paths with increasing reduced cost, and may stop only after extracting $\mu$ such paths (or once the label heap is emptied).

The set of all non-dominated labels of feasible paths is called \emph{Pareto front}.
Our \algo{\BiAstar{\mu}} returns a feasible path corresponding to each of the $\mu$ lowest-cost labels in the Pareto front.
For $\mu=\infty$, the entire Pareto front is represented.
However, adding too many paths in a single \problem{$\mu$-CSP} call
may overcrowd the LP with unnecessary variables. 
Furthermore, even if fewer than $\mu$ feasible paths exist, we still have to wait for the heap to be emptied, as deciding if all feasible paths are found is NP-hard~\cite{bokler2017}.
In pilot studies, we compared solvers using \algo{\BiAstar{\mu}}, for $\mu \in \{1,2,3,5,\infty\}$. 
Generally, using $\mu=\infty$ performs worst, while the other $\mu$-values yield very similar performances and there is no value that performs best across all instances. Overall, however, \algo{\BiAstar{3}} appears to be the most promising approach.

\mysubparagraph{Pricing: Pruning \problem{$\bm{\mu}$-CSP} calls.}
During pricing, we generally solve a \problem{$\mu$-CSP} instance for each terminal pair $\{u,v\}\in K$.
This instance is characterized by the edge costs $\pi^{uv}$ and cost limit $\sigma_{uv}$. The vector $\pi^{uv}$ contains the non-zero dual values for existing constraints \eqref{PB:edge}, and $0$ otherwise.
For each $\{u,v\} \in K$, we may store the last $(\pi^{uv},\sigma_{uv})$ that was considered without yielding a feasible path.
Subsequently, for this $\{u,v\}$, we only need to compute the current \problem{$\mu$-CSP} instance, if its new cost limit is larger or some new edge cost is lower than the old stored value. 
Otherwise, we can \emph{prune} the call without needing to perform any computation.
Clearly, the required bookkeeping is very memory consuming, but pilot studies show that it drastically reduces the number of \problem{$\mu$-CSP} calls.
Typically, only a few terminal pairs require many pricing iterations to find all required paths.

\section{Multicommodity Flow}\label{sec:AB}
The above path-based formulation \eqref{ilp:PB} can alternatively be understood to model a multicommodity flow problem (MCF, for short) with one commodity for each $\{u,v\}\in K$. A standard way to model MCF is an arc-based formulation, which is the basis for the only other ILP formulation \eqref{ilp:AB} for \problem{MWSP}~\cite{ahmed2019}.
After its short description and some straight-forward speed-up techniques, we compare \eqref{ilp:AB} and \eqref{ilp:PB}.

For each terminal pair $\{u,v\} \in K$, we route a unit flow from $u$ to $v$,
modeled via two directed \emph{flow variables} $\flowvarij$ and \flowvarji, for each $\{i,j\} \in E$, and corresponding flow conservation constraints~\eqref{ilp:AB:kirch}.
The spanner $H$ is again described by decision variables $x_e$, for all $e\in E$, which have to be $1$ if the edge carries some flow \eqref{ilp:AB:edge_arc}.
Constraints \eqref{ilp:AB:stretch} ensure that any integral $u$-$v$-flow has length at most $\alpha \cdot d_G(u, v)$. We use $\mathbb{1}_{\varphi}$ as an indicator function that is $1$ if the boolean expression $\varphi$ is true, and $0$ otherwise.
\begin{taggedsubequations}{AB}
\label{ilp:AB}
\allowdisplaybreaks
\begin{align}
\mathrlap{\text{(AB)}}&&
\text{min}  \sum_{e \in E} w_{e} x_{e} & \\
&&        \text{s.t.\ }     \sum_{j:\{i,j\} \in E} (\flowvarij - \flowvarji) &= \mathbb{1}_{i=u} - \mathbb{1}_{i=v}
                            &&  \forall \{u, v\} \in K, \forall i \in V \label{ilp:AB:kirch}\\
&&             \flowvarij + \flowvarji &\leq x_e && \forall \{u,v\} \in K, \forall e=\{i,j\} \in E \label{ilp:AB:edge_arc}\\
&&\sum_{\{i,j\} \in E} w_{\{i,j\}} (\flowvarij + \flowvarji) &\leq \alpha \cdot d_G(u,v) && \forall \{u,v\} \in K \label{ilp:AB:stretch}\\
&&             \sum_{j:\{i,j\} \in E} \flowvarij &\leq \mathbb{1}_{i \neq v} &&\forall \{u,v\} \in K, \forall i \in V
\label{ilp:AB:outflow}\\
&&              x_{e},\flowvarij, \flowvarji &\in  \{0,1\} &&  \forall \{u,v\} \in K, \forall e=\{i,j\} \in E
\end{align}%
\end{taggedsubequations}%
The non-essential constraints \eqref{ilp:AB:outflow}, with a right-hand side of $1$, are proposed in~\cite{ahmed2019} so that flow would correspond to simple paths; we use the slightly stronger right-hand side $\mathbb{1}_{i \neq v}$. %

In~\cite{ahmed2019}, several size reduction techniques are used: Firstly, they also use metrication.
Secondly, they fix \emph{unreachable} variables:
whenever, for any $\{u,v\} \in K$ and any edge direction $(i,j)$, $\{i,j\}\in E$, $d_G(u,i) + w_{\{i,j\}} + d_G(j,v) > \alpha \cdot d_G(u,v)$, they set $\flowvarij=0$.
Thirdly, they fix \emph{mandatory} variables:
if, for some $\{u,v\} \in K$ and $\{i,j\} \in E$, the edge direction $(i,j)$ is used in every  $u$-$v$-path observing the stretch constraint, then the corresponding flow variable \flowvarij (and consequently $x_{\{i,j\}}$) is fixed to~$1$.

\mysubparagraph{Our modifications.} We reduce the size of the \eqref{ilp:AB} model by using $K=E$, instead of $K=\binom{V}{2}$, analogous to the discussion in \Cref{sec:AE}.
We also provide the solver with an upper bound from the \algo{BG} heuristic.
Pilot studies show that %
our modified solver yields significantly smaller models and is able to solve larger instances than the original \eqref{ilp:AB} implementation.

\mysubparagraph{Comparing \eqref{ilp:AB} and \eqref{ilp:PB}.} In contrast to \eqref{ilp:PB}, \eqref{ilp:AB} has polynomial size and can be solved using standard B\&B frameworks. Thus, \eqref{ilp:AB} is considerably easier to implement. 
Arc-based MCF ILPs can be formulated for a wide range of problems~\cite{salimifard2022}; in practice, however, such models are often inferior to (typically also implementation-wise) more involved formulations, e.g., cut-based formulations \cite{FST07,F91,LWPKMF05,CKLM10}.
While the model size of \eqref{ilp:PB} is largely dependent on $|\PP|$, and thus on $\alpha$, \eqref{ilp:AB} is not. We can expect \eqref{ilp:PB}'s performance to degrade with increasing stretch factors, while \eqref{ilp:AB} should have similar performance across all stretches.

It is easy to see that every fractional LP-solution to \eqref{ilp:PB} can be mapped to a fractional \eqref{ilp:AB} solution with the same $x$-variable values (see, e.g., \cite{ahuja1993network}), establishing that \eqref{ilp:PB} is at least as strong as \eqref{ilp:AB}.
The inverse is not true in general as, without fixing unreachable variables, already an unweighted 4-cycle for $\alpha=2$ yields an LP-solution for \eqref{ilp:AB} with all $x_e=0.5$ (yielding a dual bound of $2$), while \eqref{ilp:PB} gives the optimum value of~$3$.
For $\alpha = 2$ on unweighted graphs, fixing unreachable variables always suffices to avoid such a discrepancy. However, even on unweighted instances with variable fixing, we observe for general $\alpha$:
\begin{theorem}
    The LP-relaxation of \eqref{ilp:AB} is strictly weaker than \eqref{ilp:PB} in general.
\end{theorem}
\begin{proof}
    Let $\alpha=5$. Consider an unweighted $K_5$, where we subdivide all edges along a Hamilton cycle once. Let $C$ be the corresponding Hamilton cycle of length~$10$. 
    For any pair of nodes $u\neq v$, every edge $e\in C$ lies on a $u$-$v$-path of length at most $5$, and thus no flow variables can be fixed. 
    Setting $x_e=0.5$ for all $e\in C$, and $x_e=0$ otherwise, allows a fractional LP-solution for \eqref{ilp:AB} of weight $5$.
Consider any fractional solution to \eqref{ilp:PB}.
For every degree-2 node $u$, the variables of its incident edges add up to at least $1$ by \eqref{PB:path}, but also to at most $1$, if the objective value were to be at most~$5$.
Thus, for every path $P$ that uses both edges adjacent to some degree-2 node, we have $y_P\leq 0.5$.
But for any pair of degree-4 nodes, there is only one path along $C$ of length at most $5$, requiring us to either increase the variable values along $C$ or have non-zero variable values for edges in $E\setminus C$.
In either case, any LP-solution to \eqref{ilp:PB} has an objective value $>5$, proving the theorem.
\end{proof}

\section{Experiments} \label{sec:Experiments}
All our implementations are freely available and will be part of the next release of the open source C++ 
library \emph{Open Graph algorithms and Datastructures Framework}~\cite{ogdf_REAL}.
All instances and detailed data of all experiments are available at~\cite{tcs.uos_REAL}.
We use the Branch-and-Cut-and-Price framework SCIP 8.0.4~\cite{scip} with CPLEX 22.1.1~\cite{cplex} as the LP solver.
All experiments are performed on an Intel Xeon Gold 6134 with 256 GB RAM under Debian 10.2 using gcc 8.3.0-6 (\texttt{-O3}). We enforce a time limit of $30$ minutes and a RAM limit of 32 GB per instance.

We consider several exact solvers. The re-implementation of the original path-based method by Sigurd and Zachariasen~\cite{Sigurd2004} is denoted \algo[Orig]{PB}. 
\algo[Top]{PB} uses our improvements, were the most promising configuration was found during pilot studies: We metricize the graph, set $K=E$, and fix mandatory variables; the set of initial paths is created by \algo{\kSP[10]+BG}; we use \algo{\BiAstar{3}} and prune \problem{$\mu$-CSP} calls.
Furthermore, we consider our improved version of the arc-based approach, denoted \algo{\ABplus}. 

\mysubparagraph{Hypotheses.}
We formulate our central research questions as falsifiable hypotheses:
\begin{itemize}[leftmargin=8mm,topsep=0mm] 
    \item[\hypShort{1}.] Even with current hardware, \algo[Orig]{PB} is unable to solve significantly larger instances compared to what was reported in 2004.
    \item[\hypShort{2}.] Our speedup techniques used on \algo[Orig]{PB} are effective and allow \algo[Top]{PB} to solve instances orders of magnitude larger than previously possible.
    \item[\hypShort{3}.] \algo[Top]{PB} is faster and able to solve larger instances than \algo{\ABplus}.
    \begin{itemize}
        \item[\hypShort{3'}.] \algo[Top]{PB} yields fewer variables than \algo{\ABplus}.
    \end{itemize}
    \item[\hypShort{4}.] Even for larger instances, \algo{BG} produces spanners with near optimum weight.
\end{itemize}

\mysubparagraph{Generated Instances.}
We place $n \in N \coloneqq \{20,50,100,200,500,1000,2000\}$ nodes uniformly at random in a unit square; previously, only graphs with $n\leq100$ could be considered.
Edges are introduced in different ways to obtain different graph classes, following the literature on evaluating spanner algorithms.
In the generation processes below, we enforce relative densities $\varrho \in \Rho \coloneqq \left\{2 \ln(n)/n, 10\%, 50\%\right\}$ or constant average node degrees $\delta \in D \coloneqq \{4,8\}$.
\begin{description}
    \item [\instance{ER}.] \emph{Erdős-Rényi} graphs were previously used in~\cite{Sigurd2004, ahmed2019, chimani2022}. For each $n \in N$ and $\varrho \in \Rho$ ($\delta \in D$), each possible edge is included with uniform probability of $\varrho$ ($\delta/(n-1)$) to get the desired relative density (average node degree, respectively).
    \item [\instance{WM}.] \emph{Waxman} graphs were previously used~\cite{Sigurd2004}. They generalize \emph{random geometric graphs}~\cite{penrose2003} (where nodes are adjacent if their euclidean distance is below some threshold) 
    and originate from applications in broadcasting~\cite{waxman1988}. Each edge is included with probability $\gamma e^{-d/(\beta L)}$, where $d$ is the distance between its endnodes and $L$ is the maximum distance between any two nodes. We use similar parameters as~\cite{Sigurd2004}, but adapted for larger $n$: For each $n \in N$ and $\varrho \in \Rho$ ($\delta \in D$), we keep $\beta = 0.14$ fixed and vary $\gamma$ depending on $n$ to achieve the desired relative density (average node degree). See~\cite{tcs.uos_REAL} for specific $\gamma$ values.
    \item [\instance{CMP}.]\emph{Complete} graphs (with relative density $\varrho=100\%$), were previously used in~\cite{Sigurd2004,chimani2022}. We consider them for all $n \in N\setminus\{2000\}$.
\end{description}
For each instance, we consider three different edge weight types $\weightSet \coloneqq \{\weight{1}, \weight{euc}, \weight{n}\}$:
\weight{1} denotes uniform weights, i.e., unweighted graphs. 
\weight{euc} considers euclidean distances.
For \weight{n}, edge weights are drawn uniformly at random from $\{1,2, \dots n\}$.
We do not consider weights $1\pm 1/3$ (a further of possibility considered in~\cite{chimani2022}), as the results therein suggest the same behavior as \weight{1}.
While \weight{1} and \weight{euc} yield metric weights, \weight{n} does not in general.

We generate 10 \instance{ER} and \instance{WM} graphs for every combination in $N \times \weightSet \times (\Rho \cup D)$, respectively. Similarly, 10 \instance{CMP} graphs are created for every combination in $(N\setminus\{2000\}) \times \weightSet$.
Note that \instance{ER} are with high probability \emph{expander graphs}~\cite{hoory2006}, i.e., simultaneously sparse and highly connected, while \instance{WM} generally are not.

\mysubparagraph{Established Instances.} We also consider preexisting (often real-world) instances.
\begin{description}
    \item[\instance{SteinLib}.] The well-known graph library~\cite{steinlib}, originally created for the Steiner tree problem. As many applications of Steiner trees and spanners overlap, it was used in~\cite{chimani2022} to evaluate spanner heuristics. We consider the 1017 graphs with at least 100 nodes. On average, they have $1398$ nodes and $10\,045$ edges; $74\%$ of graphs have a relative density $\varrho<2.5\%$.

    \item[\instance{Road}.] We consider 10 undirected US \emph{road networks}~\cite{roadUSA} with $5\,000 < |V| < 17\,000$. Their average node degrees are $2.3$--$2.9$.
    We consider two different weights: \weight{len} and \weight{t} consider the length and travel time (i.e., length$/$speed-limit) of the road segments, respectively.

    \item[\instance{Bundling}.] A recent graph drawing paper on \emph{edge-path bundling}~\cite{wallinger2023} uses spanners (computed via \algo{BG}) as their central building block.
    We can now solve their \emph{Airlines} instance ($|V|=235$, $|E|=1297$) optimally for $\alpha \in \{1.2, 1.5, 2\}$ within $4$\,s, $70$\,s, and $25$\,min, respectively.
\end{description}
For all above instance sets (both generated and established), we consider stretch factors $\alpha \in \{1.2, 1.5, 2, 3, 5\}$. 
In the unweighted case, we disregard $\alpha<2$, since there $G$ itself is the only feasible solution.
We consider large stretch factors $3$ and $5$ mainly to get a broader picture. We stress that in practice, $\alpha$ should typically be assumed to be rather small: e.g. \cite{heinrich2023} considers detours in networks beyond $\alpha = 1.5$ to be typically too long; \cite{wallinger2023} are not interested in distortions beyond $\alpha = 2.5$ in graph drawings.

\mysubparagraph{\Hyp{1}.}
In 2004,~\cite{Sigurd2004} conducted experiments on a 933 MHz Intel Pentium III with a time limit of 30 minutes.
For any $\alpha\leq2$, the largest solved \instance{WM} and \instance{ER} graphs ($\delta \in \{4, 8\}$) have $n = 64$; \instance{CMP} graphs are solved for up to $n = 50$.
On our modern machine with $\alpha=2$, \algo[Orig]{PB} is still unable to solve \instance{WM} or \instance{ER} graphs with average node degree $\delta \in \{4, 8\}$ and $n>100$, and \instance{CMP} graphs with $n > 50$. Thus, we cannot reject \hypShort{1} and take it as confirmed.

\begin{table}[tb] 
    \centering
        \caption{Top: Share (\%) of optimally solved \instance{WM} graphs. Bold font marks the better of \algo[Top]{PB} and \algo{\ABplus}. Bottom: median gap %
        of \algo{BG} to the optimum solution value (or the best found dual bound; cf.\ top). A ``--'' indicates that there are instances in that class with no non-trivial lower bound.}\label{table:shareSolvedWM}
        \smaller
    \setlength\tabcolsep{4.3pt}
    \begin{tabular}{|l*{4}{|ccc}*{2}{|cc}|} %
     \tableHead
          \hline
          \hline

        \rowcolor{lightGray}
        \algo[Top]{PB} &  $\bm{100}$ & $\bm{100}$ &  $\bm{99}$ & $\bm{100}$ &  $\bm{100}$ & $\bm{90}$ &  $\bm{100}$ & $\bm{100}$ & $\bm{75}$ &  $\bm{62}$ &  $\bm{49}$ & $25$ & $41$ & $10$ & $19$ & $8$\\
          
        \algo[Orig]{PB} &  $49$ & $34$ &   $16$ & $43$ &  $27$ & $18$ & $39$ & $16$ & $9$ & $26$ & $14$ & $8$ & $20$ & $6$ & $12$ & $7$ \\

        \rowcolor{lightGray}
      \algo{\ABplus} &  $93$ & $93$ &  $86$ & $86$ &  $81$ & $66$ &  $93$ & $93$ & $74$ &  $57$ &  $46$ & $\bm{30}$ & $\bm{100}$ & $\bm{62}$ & $\bm{48}$ & $\bm{13}$\\

        \algo[noM]{PB} &  $100$ & $100$ &  $98$ & $76$ &  $76$ & $62$ &  $100$ & $100$ & $74$ &  $62$ &  $49$ & $24$ & $41$ & $10$ & $19$ & $8$\\

        \rowcolor{lightGray}
        \algo[noE]{PB} &  $64$ & $57$ &  $44$ & $57$ &  $57$ & $40$ &  $39$ & $19$ & $12$ &  $37$ &  $17$ & $13$ & $30$ & $9$ & $15$ & $5$\\
    \hline
    \hline
    $n$ & \multicolumn{16}{c|}{Median gap of \algo{BG} to best lower bound in \% }\\
    \hline
    \hline
    \rowcolor{lightGray}
    $100$ & $0.0$ & $2.1$ & $5.0$ & $0.0$ & $2.3$ & $4.6$ & $0.5$ & $4.0$ & $6.1$ & $2.8$ & $8.6$ & $8.0$ & $5.4$ & $31.1$ & $19.8$ & $129.1$\\

    $500$ & $0.0$ & $1.3$ & $3.5$ & $0.1$ & $2.0$ & $4.9$ & $0.1$ & $3.2$ & $5.1$ & $1.4$ & $6.3$ & -- & $0.6$ & $7.3$ & $4.7$ & $90.2$\\

    \rowcolor{lightGray}
    $2000$ & $0.0$ & $0.6$ & $3.2$ & $0.0$ & $1.1$ & $3.6$ & $0.0$ & $1.1$ & $2.3$ & $0.7$ & -- & -- & $0.1$ & $1.3$ & $3.6$ & $56.4$ \\
    \hline
    \end{tabular}
\end{table}

\mysubparagraph{\Hyp{2}.}
We investigate the effect of our speedup techniques used in \algo[Top]{PB}. 
\algo[Top]{PB} is able to solve many more instances than \algo[Orig]{PB}: in particular, while the former still has very high success rates for \instance{Road} and \instance{SteinLib} (\Cref{fig:RD_W_t_time,fig:STL}), the latter solves no such instance at all. On the generated instances, the picture is analogous (see, e.g., \Cref{table:shareSolvedWM,table:medianRunningTime}): while \algo[Orig]{PB} solves only the smallest instances as discussed above, \algo[Top]{PB} typically also solves the largest instances ($n=2000$) for stretches $\alpha \leq 2$. 
As expected, instances become harder for the \eqref{ilp:PB}-approaches with increasing $\alpha$, as the variable set increases. But even for $\alpha=3$, \algo[Top]{PB} solves $70\%$ of the \weight{n}-weighted complete graphs on $1000$ nodes.
We also observe that unweighted instances are typically much more challenging than weighted instances.

We are thus interested, which of our speedup techniques enables these high success rates. Therefore,
we consider variants of \algo[Top]{PB}, where individual features are turned off (they are otherwise identical to \algo[Top]{PB}):
\algo[noM]{PB} does not use metrication, \algo[noE]{PB} uses $K=\binom{V}{2}$, \algo[noFix]{PB} does not fix mandatory variables, \algo[simpleInit]{PB} uses \algo{\kSP[1]} initialization, \algo[noPrune]{PB} does not prune \problem{$\mu$-CSP} calls,
\algo[noBiA$^*$]{PB} uses the \algo{BasicCSP} algorithm, and \algo[\BiAstar{1}]{PB} uses \algo{\BiAstar{1}}.
For brevity, we will not discuss both \instance{ER} and \instance{WM} individually, as their behaviors are very similar. Seemingly, \instance{WM} graphs are slightly more challenging to solve than \instance{ER} (\Cref{appendix:H2:table}). We define the \emph{success rate} as the percentage of instances solved to proven optimality within the time and RAM limit.
We start with considering \Cref{table:shareSolvedWM}.

\newcommand{\minipara}[1]{\textbf{#1}}
\minipara{Metrication and terminal pairs.} 
For all \weight{n}-weighted \instance{ER}, \instance{WM}, and \instance{CMP} graphs, the average remaining density (average node degree) after metrication is in $[4.0\%, 9.2\%]$ ($[3.7, 7.9]$, respectively). Thus, after metrication, all \weight{n}-weighted graphs are sparse.
On originally (nearly) metric instances, \algo[Top]{PB} and \algo[noM]{PB} show similar performances; on dense \weight{n}-weighted graphs with $\varrho \in \Rho \cup \{100\%\}$, \algo[Top]{PB} is clearly superior.
\Cref{table:shareSolvedWM} also shows that \algo[noE]{PB} is far inferior to \algo[Top]{PB}, a testament to the benefit of restricting $K$.

\begin{table}[tb] 
    \centering
        \caption{Average over the speedup factors of \algo[Top]{PB} relative to the specified algorithm, on \instance{WM} instances with $n\geq50$  solved by both solvers. Factors below $1$ indicate that \algo[Top]{PB} is slower.}\label{table:speedupFactorWM}
        \smaller
    \setlength\tabcolsep{3pt}
    \begin{tabular}{|l*{4}{|ccc}*{2}{|cc}|} %
     \tableHead

     \hline
        \hline

          \rowcolor{lightGray}
        \algo[noFix]{PB} &  $1.85$ & $1.85$ & $2.05$ & $1.96$ & $1.79$ & $1.68$ & $1.84$ & $2.14$ & $1.94$ & $1.73$ &  $1.22$ & $1.14$ & $2.91$ & $1.79$ &  $1.52$ & $1.00$ \\

        \algo[simpleInit]{PB} &  $0.83$ & $1.18$ &  $1.43$ & $0.97$ &  $1.53$ & $1.38$ &  $0.91$ & $1.63$ & $1.35$ &  $1.69$ &  $2.93$ & $1.57$ & $0.98$ & $1.30$ & $0.94$ & $1.28$\\

        \rowcolor{lightGray}
        \algo[noPrune]{PB} & $1.81$ & $2.13$ & $2.62$ & $2.00$ & $2.38$ & $2.33$ & $1.90$ & $2.80$ & $2.38$ & $2.14$ &  $1.65$ & $1.15$ & $3.10$ & $1.65$ &  $1.70$ & $1.40$ \\

      \algo[noBiA$^*$]{PB} & $0.98$ & $1.07$ & $1.96$ & $1.03$ & $1.12$ & $5.24$ & $0.97$ & $1.72$ & $2.15$ & $1.20$ &  $4.07$ & $2.45$ & $1.36$ & $1.20$ &  $1.27$ & $0.93$ \\

        \rowcolor{lightGray}
      \algo[\BiAstar{1}]{PB} & $1.00$ & $1.05$ & $1.05$ & $1.04$ & $1.01$ & $1.09$ & $1.00$ & $1.01$ & $0.85$ & $1.03$ &  $0.77$ & $0.66$ & $1.00$ & $1.01$ &  $1.01$ & $1.26$ \\

          \hline
          \hline
      \algo{\ABplus} & $13.0$ & $18.9$ & $8.07$ & $29.9$ & $34.8$ & $11.7$ & $13.6$ & $8.98$ & $2.01$ & $11.7$ & $1.72$ & $0.18$ & $0.03$ & $0.01$ & $0.05$ & $0.01$\\

          \hline
    \end{tabular}
    
\end{table}

In the context of \Cref{table:shareSolvedWM}, the success rates for the other individually deactivated speedup techniques never deviate more than $5$ ($\alpha\leq 2$) or $12$ ($\alpha>2$) percentage points from \algo[Top]{PB}.
Still, \algo[Top]{PB} achieves significant speedups against them, as shown in \Cref{table:speedupFactorWM}:

\minipara{Fixing variables.} 
As fixing mandatory variables generates no overhead, \algo[noFix]{PB} is never faster than \algo[Top]{PB}. Overall, it typically allows 1.5--2 times faster computations, in particular on sparse graphs.
The benefit decreases with increasing density or $\alpha$, as the share of mandatory path variables drops; for $\alpha=5$ almost no variables can be fixed.

\minipara{Initial path variables.}
The speedup of the proposed initialization seems to largely depend on $\alpha$. 
Clearly, the quality difference between the $1$-spanner induced by the \algo{\kSP[1]} initialization to the $\alpha$-spanner yielded by \algo{BG} grows with $\alpha$. Further, for larger $\alpha$, \algo{\kSP[10]} is able to provide more paths. 
For $\alpha \leq 2$, there are some cases where the more elaborate initialization of \algo[Top]{PB} increases the overall running time, presumably due to the fact that there are only very few feasible paths. However, for the majority of the considered classes, the \algo{\kSP[k]+BG} initialization is worthwhile and even yields close to three times faster computations than \algo[simpleInit]{PB} on dense \weight{euc}-weighted \instance{WM} graphs for $\alpha=2$.

\minipara{Pruning \problem{$\mu$-CSP} calls.}
On optimally solved instances with weights \weight{n}, \weight{euc}, and \weight{1}, we can prune $90$\%, $86$\%, and $63$\% \problem{$\mu$-CSP} calls, respectively, in the median. This yields speedup factors of up to $3.1$ compared to \algo[noPrune]{PB}. 
The share of pruned calls slightly drops with increasing $\alpha$. Surprisingly, this does not always translate into smaller speedups.

\minipara{\algo{\BiAstar{\mu}} pricing.}
Especially for $\alpha=1.2$, the speedup against \algo[noBiA$^*$]{PB} is surprisingly low and sometimes even slightly below $1$. Reasons may be that
\algo[\BiAstar{1}]{PB} has a higher setup cost than \algo{BasicCSP}, which cannot be recuperated if most necessary paths are already added during initialization,
and most \problem{$\mu$-CSP} calls are either pruned or quickly detected to be infeasible.
Still, for larger stretches, using only \algo{BasicCSP} instead of \algo{\BiAstar{3}} can lead to up to 5-fold running times.
---
On most instances, the latter also slightly outperforms pricing via \algo{\BiAstar{1}}, especially on instances that are sparse (after metrication) or for stretches $\alpha \leq2$.
\algo[\BiAstar{1}]{PB}, however, seemingly has advantages for large (practically less relevant) stretches and (even after metrication) dense graphs.
A possible explanation is based on the following observation on \algo[Top]{PB}:
A majority of the dual solution values $\pi$ (and thus edge costs in the \problem{CSP} instance) are $0$; thus in many cases the minimum cost feasible path is \emph{free}, i.e., has total cost $0$.
On optimally solved instances, the median share of added free paths for  \weight{n}-, \weight{euc}-, and \weight{1}-weighted graphs are $95$\%, $98$\%, and $83$\%, respectively; their share grows with $\alpha$ and density.
If there is a free path, there can be no other non-dominated solution for this \problem{$\mu$-CSP} call, and values $\mu>1$ only inflict overhead.

The speedup of \algo[Top]{PB} showcased against \algo[noM]{PB}, \algo[noE]{PB}, \algo[noFix]{PB}, \algo[simpleInit]{PB}, \algo[noPrune]{PB}, and \algo[noBiA$^*$]{PB} give reason to not reject \hypShort{2}.
The decision between \algo{\BiAstar{1}} or \algo{\BiAstar{3}} is less clear: it seems non-crucial and instance-property dependent.
We choose \algo{\BiAstar{3}} in \algo[Top]{PB}, as it performs slightly better on more practically relevant instances.
With no reason for rejection, we consider \hypShort{2} confirmed.

\mysubparagraph{\Hyp{3}.}
We compare \algo[Top]{PB} and \algo{\ABplus}. 
Our experiments show a clear divide between (non-uniformly) weighted and unweighted (\weight{1}) graphs, so we discuss them separately.

\minipara{Weighted graphs.}
\Cref{table:shareSolvedWM,table:speedupFactorWM} show the success rates and relative speeds on \instance{WM} graphs (again, their behavior on \instance{ER} graphs is analogous): \algo[Top]{PB} can solve significantly more instances than \algo{\ABplus} and achieves significant speedups, especially on graphs that are sparse (after metrication) and for stretches $\alpha \leq2$. 
For \weight{euc}, the advantage of \algo[Top]{PB} seems to shrink with increasing $\alpha$.
For $\alpha=5$, \algo{\ABplus} is able to solve more instances than \algo[Top]{PB} regardless of weight or density; however, 
neither algorithm can solve such instances with $n>200$.
\Cref{table:medianRunningTime} shows the median running times for different graph sizes and edge weights on \instance{WM} and \instance{CMP} graphs for $\alpha=1.5$.
Typically, \algo[Top]{PB} is orders of magnitudes faster, and always requires less than 1 minute for any graph that is sparse (after metrication).
\algo{\ABplus} is only beneficial on dense ($\varrho=10\%$) \weight{euc}-weighted \instance{WM} graphs; however, on those graphs, again neither algorithm can go beyond 200 nodes.

\begin{table}[tb] 
    \centering
        \caption{Median running time (in sec) for \instance{WM} (specified by $\delta$ or $\varrho$) and \instance{CMP} graphs with $\alpha=1.5$. 
        ``--'' indicate success rates below $100\%$: $30\%$ for \algo[Top]{PB}, and $0\%$ for \algo{\ABplus}, respectively}\label{table:medianRunningTime}
        \smaller
    \setlength\tabcolsep{3.4pt}
    \begin{tabular}{|l*{2}{|rr|rr|rr|rr|rr|rr|rr}|} %
      \hline
           weight &  \multicolumn{8}{c|}{\weight{n}} &  \multicolumn{6}{c|}{\weight{euc}} \\
          \hline
        density &  \multicolumn{2}{c|}{$\delta=4$} &  \multicolumn{2}{c|}{$\delta=8$} &  \multicolumn{2}{c|}{$\varrho=10\%$} &  \multicolumn{2}{c|}{\instance{CMP}} &  \multicolumn{2}{c|}{$\delta=4$} & \multicolumn{2}{c|}{$\delta=8$} & \multicolumn{2}{c|}{$\varrho=10\%$} \\
          \hline
        $n$ & $1000$ & $2000$ & $1000$ & $2000$ & $1000$ & $2000$ & $500$ & $1000$ & 
        $1000$ & $2000$ & $1000$ & $2000$ & $100$ & $200$ \\
        \hline\hline
          \rowcolor{lightGray}
        \algo[Top]{PB} &  $2.6$ & $11.4$ & $5.7$ & $31.6$ & $7.7$ & $45.6$ & $2.9$ &$15.6$ &
         $2.9$ & $11.6$ & $9.2$ & $45.2$ & $0.5$ & --\\
          \hline          
      \algo{\ABplus} & $35.9$ & $172.9$ & $672.0$ & -- & $1258.2$ & -- &  $328.7$ & -- & $44.8$ & $225.8$ & $818.9$ & -- & $3.7$ & $258.9$\\     
     \hline
    \end{tabular}
\end{table}

\Cref{fig:ER_euc_ln_time} shows the median running time of \algo[Top]{PB} and \algo{\ABplus} for \weight{euc}-weighted \instance{ER} graphs with $\varrho = 2\ln (n)/n$, for varying sizes and stretches. Running times of non-solved instances are interpreted as $\infty$;
consequently, we only show data for success rates above $50\%$.
For $\alpha=5$, no algorithm can reliably solve graphs with more than 20 nodes.
For all other $\alpha$, \algo[Top]{PB} solves significantly larger instances than \algo{\ABplus}.
For $\alpha\leq2$, \algo[Top]{PB} reliably solves instances with 2000 nodes, while \algo{\ABplus} can only solve instances with $n\leq 500$.
For all $\alpha\leq 2$, median running times of \algo[Top]{PB} are orders of magnitude lower than those of \algo{\ABplus}.

\Cref{fig:RD_W_t_time} shows the running times for the real-world \weight{t}-weighted \instance{Road} instances; results for \weight{len} are similar.
For $\alpha \leq 2$, \algo[Top]{PB} solves all instances optimally, even the one with $16\,874$ nodes, and is orders of magnitudes faster than \algo{\ABplus}.
In contrast to \algo[Top]{PB}, \algo{\ABplus}---which has similar running times for all $\alpha\leq3$---is able to solve some small instances for $\alpha=3$.
Neither solved any \instance{Road} graph for $\alpha=5$.
Consider the established \instance{SteinLib}, see \Cref{fig:STL} for $\alpha=1.5$.  
Generally, most groups behave similarly to sparse \weight{n}- or \weight{euc}-weighted graphs.
Except for \emph{Group ST} instances, \algo[Top]{PB} performs better than \algo{\ABplus} and never solves fewer instances.
Interestingly, the \emph{Cross-Grid} graphs are unweighted but---in contrast to most unweighted graphs, discussed below---\algo[Top]{PB} solves them for $\alpha=1.5$ over 10 time \emph{faster} than \algo{\ABplus}.

\begin{figure}
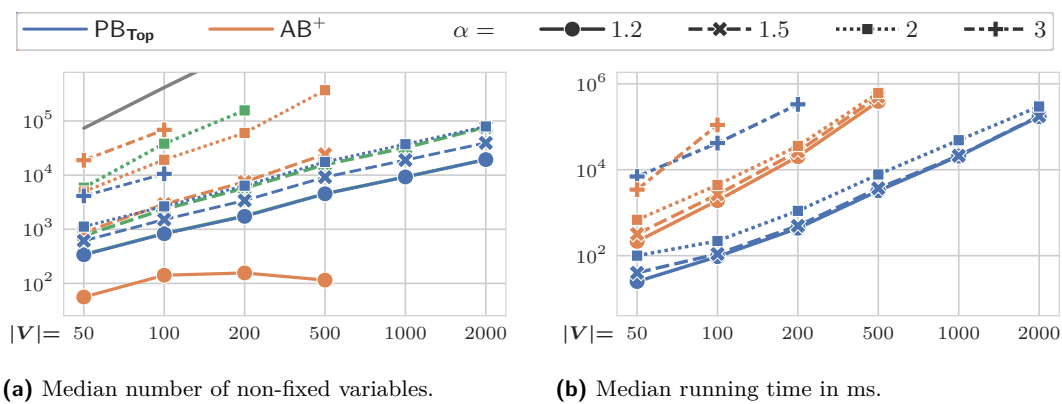

	\centering
\input{figures/legend.pgf}
\begin{subfigure}[t]{.48\textwidth}
\input{figures/ER_euc_ln_numVars.pgf}
	\caption{Median number of non-fixed variables.}
  \label{fig:ER_euc_ln_numVars}
	\end{subfigure}
 \hfill
\begin{subfigure}[t]{.48\textwidth}
\input{figures/ER_euc_ln_time.pgf}
  \caption{Median running time in ms.}
  \label{fig:ER_euc_ln_time}
 \end{subfigure}
	\caption{Comparing \algo[Top]{PB} and \algo{\ABplus} on \weight{euc}-weighted \instance{ER} graphs with $\varrho = 2 \ln(n)/n$. In \textsf{\textbf{(a)}}, the gray line represents the number of variables in the model \eqref{ilp:AB} (without fixing); the green lines (dependent on $\alpha$) give the upper bound \algo{PB-UB} of non-fixed variables for \algo[Top]{PB}, if one would consider all path variables in the full ILP (\ref{ilp:PB}). For $\alpha=1.2$, \algo[Top]{PB} attains this upper bound.
 \label{fig:with_legend}}
\end{figure}
\minipara{Unweighted graphs.}
As witnessed for \instance{WM} in \Cref{table:shareSolvedWM,table:medianRunningTime} (and analogously for \instance{ER}), unweighted instances are generally much harder for both \algo{\ABplus} and \algo[Top]{PB}, and we can solve only fewer and smaller graphs. On those graphs, however, \algo{\ABplus} typically significantly outperforms \algo[Top]{PB}.
\algo[Top]{PB} cannot solve any generated instance with $n\geq 100$ for $\varrho \in \Rho\cup\{100\%\}$. On unweighted instances, $\alpha$ needs to be integer, but only few instances were solved by either algorithm for $\alpha\geq3$. So we concentrate on $\alpha=2$ in the following:
Let \algo{PB-UB} denote the number of variables in the full \eqref{ilp:PB}-model, but after fixing variables. This yields an upper bound on \algo[Top]{PB}'s number of variables (which would, e.g., also be attained by the brute-force initialization discussed in~\Cref{sec:AE}).
We observe that \algo[Top]{PB} generates close to \algo{PB-UB} many variables, whereas \algo{\ABplus}, thanks to fixing unreachable and mandatory variables, requires only comparably few. The former seems to struggle identifying the best of the large set of paths with identical length.
Furthermore, we observe that SCIP's default cut generators successfully separate several additional constraints to the \eqref{ilp:AB}-model, whereas they find none for the \eqref{ilp:PB}-model. This leads to the effect that, for all densities except $\varrho\in\{50\%,100\%\}$, \algo{\ABplus} solves the majority of instances at the root B\&B node. \algo[Top]{PB} requires a median number of $36\,000$ B\&B nodes already for $\delta=8$ and $n=50$.

Overall, we have to reject \hyp{3} on unweighted instances, very dense \weight{euc}-weighted graphs, and large (arguably less practically relevant) stretches $\alpha\geq 5$.
For all other instances, we cannot reject \hypShort{3}, as \algo[Top]{PB} is almost always significantly faster and able to solve larger instances than \algo{\ABplus}.
The hypothesis looks particularly strong for stretches $\alpha\leq 2$.

\mysubparagraph{\Hyp{3'}.}
The rationale for this hypothesis is that in the aforementioned cases where cut-based ILPs dominate arc-based MCF formulations, the former typically yield practically smaller models. Thus, it seems natural that \algo[Top]{PB}'s strength could be due to a smaller set of variables, compared to the $\mathcal{O}(n^4)$ of \eqref{ilp:AB}. However, the picture seems much more complicated here:
As mentioned above, on unweighted graphs (with $\alpha=2$ as the only statistically significant case) the LPs of \algo{\ABplus} hold fewer variables than \algo[Top]{PB}. Consider (non-uniformly) weighted graphs:
As a representative example, \Cref{fig:ER_euc_ln_numVars} shows the median number of non-fixed variables per algorithm, whenever all respective \weight{euc}-weighted \instance{ER} instances with $\varrho=2\ln(n)/n$ are solved; similar results hold for the other non-uniformly weighted generated instances.
For both algorithms, this number grows with $\alpha$.
However, for \algo{\ABplus} and $\alpha\leq2$, these differences in model size are \emph{not} reflected in the eventual running times, see \Cref{fig:ER_euc_ln_time}.
In fact, for $\alpha=1.2$, \algo{\ABplus} yields the smallest median model size, but its median running time is still orders of magnitudes larger than \algo[Top]{PB}'s.
\algo{\ABplus} spends a significant amount of time fixing variables, which dominates the running time for $\alpha < 2$.
However, only for $\alpha=1.2$, it yields fewer unfixed variables than \algo[Top]{PB}; otherwise \algo{\ABplus} typically requires significantly more variables.
Sometimes, it even yields models larger than \algo{PB-UB}.
Consistently, on the very sparse \instance{Road} instances (\Cref{fig:RD_W_t_numVars}), \algo[Top]{PB}'s variable number is significantly larger than \algo{\ABplus} (and surprisingly independent of $\alpha$ for all $\alpha\leq 2$).
Also on \instance{SteinLib}, \algo{\ABplus} typically requires less variables than \algo[Top]{PB}.
In any case, \algo{\ABplus}'s strive for a small variable set is typically too expensive to attain competitive running times---in particular since the reduction is very successful only for small $\alpha$, where \algo[Top]{PB} is still faster.

Overall, while there are several cases where \algo[Top]{PB} requires less variables than \algo{\ABplus}, we overall reject \hypShort{3'} in its generality. Certainly, it cannot be used as the hypothesized simple explanation for \algo[Top]{PB}'s strong practical performance. Its performance benefit over \algo{\ABplus} seems to be rooted in the latter's drawback: \eqref{ilp:AB}'s model size, and the very time-consuming preprocessing required to mitigate its effect (if successful at all).

\begin{figure}
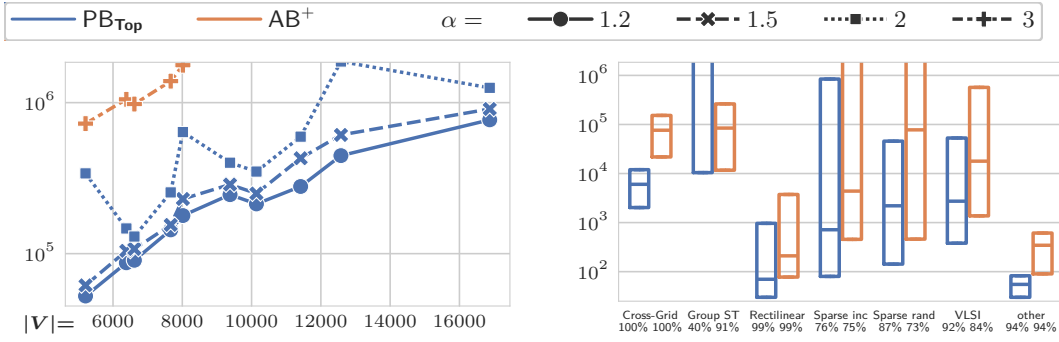

\input{figures/legend.pgf}
\begin{subfigure}[t]{.48\textwidth}\centering
   \input{figures/RD_W_t_time.pgf}
 \caption{\instance{Road} instances with \weight{t} (i.e., travel times). \algo{\ABplus} is visually indentical for all $\alpha\in\{1.2,1.5,2,3\}$.}
   \label{fig:RD_W_t_time}
 \end{subfigure}
 \hfill
\begin{subfigure}[t]{.48\textwidth}\centering
   \input{figures/STL.pgf}
 \caption{\instance{SteinLib}, grouped by type, for $\alpha=1.5$.}
   \label{fig:STL}
 \end{subfigure}
   \caption{Running time (in ms; always on the vertical log-scale axis) of \algo[Top]{PB} and \algo{\ABplus}. The legend of \Cref{fig:with_legend} applies. On the horizontal axis of (b), we give the corresponding success rates. Observe the time limit of $1.8 \cdot 10^6$ms.} %
   \label{fig:rt_specialgraphs}
\end{figure}

\mysubparagraph{\Hyp{4}.}\label{para:H4} The lower part of \Cref{table:shareSolvedWM} lists the median gap (\%) of \algo{BG} to the optimum solution (or best lower bound), for \instance{WM} graphs with varying $n\in N$. Results are more significant for high success rates. %
Generally, median gaps grow with $\alpha$ and are significantly larger on unweighted graphs.
In~\cite{Sigurd2004}, they observed decreasing gaps for increasing $n \leq 64$.
Interestingly, our experiments on significantly larger graphs show that both the median gap and even the corresponding interquartile range decreases with increasing $n$ (\Cref{appendix:H4:table1bottom}).

For $\alpha\in\{1.2, 1.5, 2\}$, the gaps on \emph{Airlines}~\cite{wallinger2023} are $0.8\%$, $5.2\%$, and $5.8\%$.
On both \weight{t}- and \weight{len}-weighted \instance{Road} graphs, the median gaps are $0.0\%$, $0.1\%$, and $0.5\%$, respectively (\Cref{appendix:H4:tableSpecial}).
On all \instance{SteinLib} groups in \Cref{fig:STL}, except for \emph{other}, \algo{BG} yields a median gap of $0.0\%$ for $\alpha \in \{1.2, 1.5\}$.
For $\alpha=2$, the gap is below $0.9\%$ on all groups except \emph{VLSI} ($15.7\%$) and unweighted \emph{Cross-Grid} ($84.8\%$).
Overall, we consider \hypShort{4} confirmed for weighted graphs.

\mysubparagraph{Summary.}
We improved both previously known exact MWSP algorithms;
in particular, our speedup techniques for the path-based column generation approach enable us to solve instances orders of magnitude larger than previous attempts.
On most instances, this approach is superior to the arc-based model. Exceptions are instances which are structurally challenging for both approaches (but, according to literature, less practically relevant): unweighted graphs and instances with very large stretch factors $\alpha \geq 5$. Lastly, we used the newly found lower bounds to evaluate the quality of the strongest and most prominent \emph{basic greedy} heuristic: Even for large instances, it produces near-optimal spanners, with the quality surprisingly improving with instance size. Unsurprisingly, it is also challenged by unweighted instances, as its strategy to consider edges in order of increasing weight is ineffective.

In the future, one may try to further improve the \algo{\BiAstar{\mu}} algorithm by exploiting the many $0$-cost edges, or by identifying further strengthening constraint classes.

\newpage
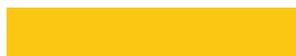
\bibliography{bibliography.bib}

\begin{thebibliography}{10}

\bibitem{ahmed2020}
Reyan Ahmed, Greg Bodwin, Faryad~Darabi Sahneh, Keaton Hamm, Mohammad
  Javad~Latifi Jebelli, Stephen Kobourov, and Richard Spence.
\newblock {Graph spanners: A tutorial review}.
\newblock {\em Computer Science Review}, 37, 2020.
\newblock \href {https://doi.org/10.1016/j.cosrev.2020.100253}
  {\path{doi:10.1016/j.cosrev.2020.100253}}.

\bibitem{ahmed2019}
Reyan Ahmed, Keaton Hamm, Mohammad~Javad Latifi~Jebelli, Stephen Kobourov,
  Faryad~Darabi Sahneh, and Richard Spence.
\newblock Approximation algorithms and an integer program for multi-level graph
  spanners.
\newblock In {\em Analysis of Experimental Algorithms: Special Event, SEA$^2$
  2019, Kalamata, Greece, June 24-29, 2019}, pages 541--562. Springer, 2019.
\newblock \href {https://doi.org/10.1007/978-3-030-34029-2_35}
  {\path{doi:10.1007/978-3-030-34029-2_35}}.

\bibitem{ahuja1993network}
Ravindra~K Ahuja, Thomas~L Magnanti, and James~B Orlin.
\newblock {\em Network flows: Theory, applications and algorithms}.
\newblock Englewood Cliffs, New Jersey, USA Arrow, KJ: Prentice-Hall, 1993.

\bibitem{alstrup2022}
Stephen Alstrup, S{\o}ren Dahlgaard, Arnold Filtser, Morten St{\"o}ckel, and
  Christian Wulff-Nilsen.
\newblock Constructing light spanners deterministically in near-linear time.
\newblock {\em Theoretical Computer Science}, 907:82--112, 2022.
\newblock \href {https://doi.org/10.1016/j.tcs.2022.01.021}
  {\path{doi:10.1016/j.tcs.2022.01.021}}.

\bibitem{althofer1993}
Ingo Alth{\"o}fer, Gautam Das, David Dobkin, Deborah Joseph, and Jos{\'e}
  Soares.
\newblock On {S}parse {S}panners of {W}eighted {G}raphs.
\newblock {\em Discrete \& Computational Geometry}, 9(1):81--100, 1993.
\newblock \href {https://doi.org/10.1007/BF02189308}
  {\path{doi:10.1007/BF02189308}}.

\bibitem{baldacci2012}
Roberto Baldacci, Aristide Mingozzi, and Roberto Roberti.
\newblock New state-space relaxations for solving the traveling salesman
  problem with time windows.
\newblock {\em INFORMS Journal on Computing}, 24(3):356--371, 2012.
\newblock \href {https://doi.org/10.1287/ijoc.1110.0456}
  {\path{doi:10.1287/ijoc.1110.0456}}.

\bibitem{barnhart1998}
Cynthia Barnhart, Natashia~L Boland, Lloyd~W Clarke, Ellis~L Johnson, George~L
  Nemhauser, and Rajesh~G Shenoi.
\newblock Flight string models for aircraft fleeting and routing.
\newblock {\em Transportation science}, 32(3):208--220, 1998.
\newblock \href {https://doi.org/10.1287/trsc.32.3.208}
  {\path{doi:10.1287/trsc.32.3.208}}.

\bibitem{baswana2007}
Surender Baswana and Sandeep Sen.
\newblock A {S}imple and {L}inear {T}ime {R}andomized {A}lgorithm for computing
  sparse spanners in weighted graphs.
\newblock {\em Random Structures \& Algorithms}, 30(4):532--563, 2007.
\newblock \href {https://doi.org/10.1002/rsa.20130}
  {\path{doi:10.1002/rsa.20130}}.

\bibitem{berman2013}
Piotr Berman, Arnab Bhattacharyya, Konstantin Makarychev, Sofya Raskhodnikova,
  and Grigory Yaroslavtsev.
\newblock Approximation algorithms for spanner problems and directed steiner
  forest.
\newblock {\em Information and Computation}, 222:93--107, 2013.
\newblock \href {https://doi.org/10.1016/j.ic.2012.10.007}
  {\path{doi:10.1016/j.ic.2012.10.007}}.

\bibitem{scip}
Ksenia Bestuzheva, Mathieu Besan\c{c}on, Wei-Kun Chen, Antonia Chmiela, Tim
  Donkiewicz, Jasper van Doornmalen, Leon Eifler, Oliver Gaul, Gerald Gamrath,
  Ambros Gleixner, Leona Gottwald, Christoph Graczyk, Katrin Halbig, Alexander
  Hoen, Christopher Hojny, Rolf van~der Hulst, Thorsten Koch, Marco
  L\"{u}bbecke, Stephen~J. Maher, Frederic Matter, Erik M\"{u}hmer, Benjamin
  M\"{u}ller, Marc~E. Pfetsch, Daniel Rehfeldt, Steffan Schlein, Franziska
  Schl\"{o}sser, Felipe Serrano, Yuji Shinano, Boro Sofranac, Mark Turner,
  Stefan Vigerske, Fabian Wegscheider, Philipp Wellner, Dieter Weninger, and
  Jakob Witzig.
\newblock Enabling research through the {SCIP} optimization suite 8.0.
\newblock {\em ACM Trans. Math. Softw.}, 49(2), 2023.
\newblock \href {https://doi.org/10.1145/3585516} {\path{doi:10.1145/3585516}}.

\bibitem{bhattacharyya2012}
Arnab Bhattacharyya, Elena Grigorescu, Kyomin Jung, Sofya Raskhodnikova, and
  David~P Woodruff.
\newblock Transitive-closure spanners.
\newblock {\em SIAM Journal on Computing}, 41(6):1380--1425, 2012.

\bibitem{roadUSA}
Geoff Boeing.
\newblock Street network models and measures for every u.s. city, county,
  urbanized area, census tract, and zillow-defined neighborhood.
\newblock {\em Urban Science}, 3(1), 2019.
\newblock URL: \url{https://www.mdpi.com/2413-8851/3/1/28}, \href
  {https://doi.org/10.3390/urbansci3010028}
  {\path{doi:10.3390/urbansci3010028}}.

\bibitem{bokler2017}
Fritz B{\"o}kler, Matthias Ehrgott, Christopher Morris, and Petra Mutzel.
\newblock Output-sensitive complexity of multiobjective combinatorial
  optimization.
\newblock {\em Journal of Multi-Criteria Decision Analysis}, 24(1-2):25--36,
  2017.

\bibitem{tcs.uos_REAL}
Fritz Bökler, Markus Chimani, Henning Jasper, and Mirko~H. Wagner.
\newblock Experimental data.
\newblock \url{https://tcs.uos.de/research/spanner}, 2024.

\bibitem{CAI1994187}
Leizhen Cai.
\newblock {NP}-completeness of minimum spanner problems.
\newblock {\em Discrete Applied Mathematics}, 48(2):187--194, 1994.
\newblock \href {https://doi.org/10.1016/0166-218X(94)90073-6}
  {\path{doi:10.1016/0166-218X(94)90073-6}}.

\bibitem{CaiKeil1994}
Leizhen Cai and Mark Keil.
\newblock Spanners in graphs of bounded degree.
\newblock {\em Networks}, 24(4):233--249, 1994.
\newblock \href {https://doi.org/10.1002/net.3230240406}
  {\path{doi:10.1002/net.3230240406}}.

\bibitem{chandra1992}
Barun Chandra, Gautam Das, Giri Narasimhan, and Jos{\'e} Soares.
\newblock New sparseness results on graph spanners.
\newblock In {\em Proc. SoCG 1992}, pages 192--201. ACM, 1992.
\newblock \href {https://doi.org/10.1145/142675.142717}
  {\path{doi:10.1145/142675.142717}}.

\bibitem{ogdf_REAL}
Markus Chimani, Carsten Gutwenger, Michael J{\"u}nger, Gunnar~W Klau, Karsten
  Klein, and Petra Mutzel.
\newblock {The Open Graph Drawing Framework (OGDF)}.
\newblock In Roberto Tamassia, editor, {\em {Handbook of Graph Drawing and
  Visualization}}, chapter~17. CRC press, 2014.

\bibitem{CKLM10}
Markus Chimani, Maria Kandyba, Ivana Ljubić, and Petra Mutzel.
\newblock Orientation-based models for \{0,1,2\}-survivable network design:
  Theory and practice.
\newblock 124(1):413--439.
\newblock \href {https://doi.org/10.1007/s10107-010-0375-5}
  {\path{doi:10.1007/s10107-010-0375-5}}.

\bibitem{chimani2022}
Markus Chimani and Finn Stutzenstein.
\newblock Spanner approximations in practice.
\newblock In Shiri Chechik, Gonzalo Navarro, Eva Rotenberg, and Grzegorz
  Herman, editors, {\em Algorithms--ESA 2022}, volume 244 of {\em LIPIcs},
  pages 37:1--37:15, 2022.
\newblock \href {https://doi.org/10.4230/LIPIcs.ESA.2022.37}
  {\path{doi:10.4230/LIPIcs.ESA.2022.37}}.

\bibitem{conforti2014}
Michele Conforti, G{\'e}rard Cornu{\'e}jols, Giacomo Zambelli, Michele
  Conforti, G{\'e}rard Cornu{\'e}jols, and Giacomo Zambelli.
\newblock {\em Integer programming models}.
\newblock Springer, 2014.

\bibitem{cplex}
IBM~ILOG Cplex.
\newblock V12. 1: User’s manual for cplex.
\newblock {\em International Business Machines Corporation}, 46(53):157, 2009.
\newblock URL: \url{https://www.ibm.com/analytics/cplex-optimizer}.

\bibitem{desrochers1988}
Martin Desrochers and Fran{\c{c}}ois Soumis.
\newblock A generalized permanent labelling algorithm for the shortest path
  problem with time windows.
\newblock {\em INFOR: Information Systems and Operational Research},
  26(3):191--212, 1988.
\newblock \href {https://doi.org/10.1080/03155986.1988.11732063}
  {\path{doi:10.1080/03155986.1988.11732063}}.

\bibitem{ebert1987}
J{\"u}rgen Ebert.
\newblock A versatile data structure for edge-oriented graph algorithms.
\newblock {\em Communications of the ACM}, 30(6):513--519, 1987.
\newblock \href {https://doi.org/10.1145/214762.214769}
  {\path{doi:10.1145/214762.214769}}.

\bibitem{elkin2017}
Michael Elkin and Ofer Neiman.
\newblock Efficient algorithms for constructing very sparse spanners and
  emulators.
\newblock In {\em SODA 2017}, pages 652--669. ACM SIAM, 2017.
\newblock \href {https://doi.org/10.1137/1.9781611974782.41}
  {\path{doi:10.1137/1.9781611974782.41}}.

\bibitem{elkin2015}
Michael Elkin, Ofer Neiman, and Shay Solomon.
\newblock Light spanners.
\newblock {\em SIAM Journal on Discrete Mathematics}, 29(3):1312--1321, 2015.
\newblock \href {https://doi.org/10.1137/140979538}
  {\path{doi:10.1137/140979538}}.

\bibitem{elkin2016}
Michael Elkin and Shay Solomon.
\newblock Fast constructions of lightweight spanners for general graphs.
\newblock {\em ACM Transactions on Algorithms (TALG)}, 12(3):1--21, 2016.
\newblock \href {https://doi.org/10.1145/2836167} {\path{doi:10.1145/2836167}}.

\bibitem{festa2015}
Paola Festa.
\newblock {Constrained shortest path problems: state-of-the-art and recent
  advances}.
\newblock In {\em International Conference on Transparent Optical Networks
  (ICTON)}, pages 1--17. IEEE, 2015.
\newblock \href {https://doi.org/10.1109/ICTON.2015.7193456}
  {\path{doi:10.1109/ICTON.2015.7193456}}.

\bibitem{F91}
Matteo Fischetti.
\newblock Facets of two {{Steiner}} arborescence polyhedra.
\newblock 51(1):401--419.
\newblock \href {https://doi.org/10.1007/BF01586946}
  {\path{doi:10.1007/BF01586946}}.

\bibitem{FST07}
Matteo Fischetti, Juan-José Salazar-Gonzalez, and Paolo Toth.
\newblock The {{Generalized Traveling Salesman}} and {{Orienteering Problems}}.
\newblock In Gregory Gutin and Abraham~P. Punnen, editors, {\em The {{Traveling
  Salesman Problem}} and {{Its Variations}}}, pages 609--662. Springer US.
\newblock \href {https://doi.org/10.1007/0-306-48213-4_13}
  {\path{doi:10.1007/0-306-48213-4_13}}.

\bibitem{Garey1979}
Michael~Robert Garey and David~Stifler Johnson.
\newblock {\em Computers and {I}ntractability: {A} {G}uide to the {T}heory of
  {NP}-{C}ompleteness}.
\newblock A Series of Books in the Mathematical Sciences. W. H. Freeman and
  Company, New York, 1979.

\bibitem{heinrich2023}
Irene Heinrich, Olli Herrala, Philine Schiewe, and Topias Terho.
\newblock {Using Light Spanning Graphs for Passenger Assignment in Public
  Transport}.
\newblock In Daniele Frigioni and Philine Schiewe, editors, {\em 23rd Symposium
  on Algorithmic Approaches for Transportation Modelling, Optimization, and
  Systems (ATMOS 2023)}, volume 115 of {\em Open Access Series in Informatics
  (OASIcs)}, pages 2:1--2:16, Dagstuhl, Germany, 2023. Schloss Dagstuhl --
  Leibniz-Zentrum f{\"u}r Informatik.
\newblock URL:
  \url{https://drops-dev.dagstuhl.de/entities/document/10.4230/OASIcs.ATMOS.2023.2},
  \href {https://doi.org/10.4230/OASIcs.ATMOS.2023.2}
  {\path{doi:10.4230/OASIcs.ATMOS.2023.2}}.

\bibitem{hoory2006}
Shlomo Hoory, Nathan Linial, and Avi Wigderson.
\newblock Expander graphs and their applications.
\newblock {\em Bulletin of the American Mathematical Society}, 43(4):439--561,
  2006.

\bibitem{jha2013}
Madhav Jha and Sofya Raskhodnikova.
\newblock Testing and reconstruction of lipschitz functions with applications
  to data privacy.
\newblock {\em SIAM Journal on Computing}, 42(2):700--731, 2013.
\newblock \href {https://doi.org/10.1109/FOCS.2011.13}
  {\path{doi:10.1109/FOCS.2011.13}}.

\bibitem{kobayashi2018}
Yusuke Kobayashi.
\newblock {NP}-hardness and fixed-parameter tractability of the minimum spanner
  problem.
\newblock {\em Theoretical Computer Science}, 746:88--97, 2018.
\newblock \href {https://doi.org/10.1016/j.tcs.2018.06.031}
  {\path{doi:10.1016/j.tcs.2018.06.031}}.

\bibitem{steinlib}
Thorsten Koch, Alexander Martin, and Stefan Vo{\ss}.
\newblock Steinlib: An updated library on {S}teiner tree problems in graphs.
\newblock {T}ech {R}eport {ZIB-R}eport 00-37, Konrad-Zuse-Zentrum f{\"u}r
  Informationstechnik Berlin, 2000.
\newblock URL: \url{http://steinlib.zib.de/steinlib.php}.

\bibitem{liu2001}
Gang Liu and KG~Ramakrishnan.
\newblock A* prune: an algorithm for finding k shortest paths subject to
  multiple constraints.
\newblock In {\em Proceedings IEEE INFOCOM 2001. Conference on Computer
  Communications. Twentieth Annual Joint Conference of the IEEE Computer and
  Communications Society (Cat. No. 01CH37213)}, volume~2, pages 743--749. IEEE,
  2001.

\bibitem{LWPKMF05}
Ivana Ljubic, Ren{\'{e}} Weiskircher, Ulrich Pferschy, Gunnar~W. Klau, Petra
  Mutzel, and Matteo Fischetti.
\newblock Solving the prize-collecting steiner tree problem to optimality.
\newblock In {\em Proceedings of the Seventh Workshop on Algorithm Engineering
  and Experiments and the Second Workshop on Analytic Algorithmics and
  Combinatorics, Vancouver, BC, Canada}, pages 68--76. {SIAM}, 2005.

\bibitem{martins1984}
Ernesto Queiros~Vieira Martins.
\newblock On a multicriteria shortest path problem.
\newblock {\em European Journal of Operational Research}, 16(2):236--245, 1984.

\bibitem{mingozzi1999}
Aristide Mingozzi, Marco~A Boschetti, Salvatore Ricciardelli, and Lucio Bianco.
\newblock A set partitioning approach to the crew scheduling problem.
\newblock {\em Operations Research}, 47(6):873--888, 1999.
\newblock \href {https://doi.org/10.1287/opre.47.6.873}
  {\path{doi:10.1287/opre.47.6.873}}.

\bibitem{peleg1989}
David Peleg and {Alejandro A.} Sch{\"a}ffer.
\newblock Graph spanners.
\newblock {\em Journal of graph theory}, 13(1):99--116, 1989.
\newblock \href {https://doi.org/10.1002/jgt.3190130114}
  {\path{doi:10.1002/jgt.3190130114}}.

\bibitem{Ullman1989}
David Peleg and Jeffrey~D. Ullman.
\newblock An optimal synchronizer for the hypercube.
\newblock In {\em Proc. PODC 1987}, pages 77--85. ACM, 1987.
\newblock \href {https://doi.org/10.1145/41840.41847}
  {\path{doi:10.1145/41840.41847}}.

\bibitem{penrose2003}
Mathew Penrose.
\newblock {\em Random geometric graphs}, volume~5.
\newblock OUP Oxford, 2003.

\bibitem{pugliese2013}
Luigi Di~Puglia Pugliese and Francesca Guerriero.
\newblock A survey of resource constrained shortest path problems: Exact
  solution approaches.
\newblock {\em Networks}, 62(3):183--200, 2013.
\newblock \href {https://doi.org/10.1002/net.21511}
  {\path{doi:10.1002/net.21511}}.

\bibitem{roditty2011}
Liam Roditty and Uri Zwick.
\newblock On dynamic shortest paths problems.
\newblock {\em Algorithmica}, 61:389--401, 2011.
\newblock \href {https://doi.org/10.1007/s00453-010-9401-5}
  {\path{doi:10.1007/s00453-010-9401-5}}.

\bibitem{salimifard2022}
Khodakaram Salimifard and Sara Bigharaz.
\newblock The multicommodity network flow problem: state of the art
  classification, applications, and solution methods.
\newblock {\em Operational Research}, 22(2):1--47, 2022.
\newblock \href {https://doi.org/10.1007/s12351-020-00564-8}
  {\path{doi:10.1007/s12351-020-00564-8}}.

\bibitem{shpungin2010}
Hanan Shpungin and Michael Segal.
\newblock Near-optimal multicriteria spanner constructions in wireless ad hoc
  networks.
\newblock {\em IEEE/ACM Transactions on Networking}, 18(6):1963--1976, 2010.
\newblock \href {https://doi.org/10.1109/INFCOM.2009.5061918}
  {\path{doi:10.1109/INFCOM.2009.5061918}}.

\bibitem{Sigurd2004}
Mikkel Sigurd and Martin Zachariasen.
\newblock Construction of minimum-weight spanners.
\newblock In {\em Algorithms--ESA 2004}, pages 797--808. Springer, 2004.
\newblock \href {https://doi.org/10.1007/978-3-540-30140-0_70}
  {\path{doi:10.1007/978-3-540-30140-0_70}}.

\bibitem{thomas2019}
Barrett~W Thomas, Tobia Calogiuri, and Mike Hewitt.
\newblock An exact bidirectional {A}$^{*}$ approach for solving
  resource-constrained shortest path problems.
\newblock {\em Networks}, 73(2):187--205, 2019.
\newblock \href {https://doi.org/10.1002/net.21856}
  {\path{doi:10.1002/net.21856}}.

\bibitem{wallinger2023}
Markus Wallinger, Daniel Archambault, David Auber, Martin N{\"{o}}llenburg, and
  Jaakko Peltonen.
\newblock Faster edge-path bundling through graph spanners.
\newblock {\em Comput. Graph. Forum}, 42(6), 2023.
\newblock \href {https://doi.org/10.1111/CGF.14789}
  {\path{doi:10.1111/CGF.14789}}.

\bibitem{waxman1988}
Bernard~M. Waxman.
\newblock Routing of multipoint connections.
\newblock {\em {IEEE} J. Sel. Areas Commun.}, 6(9):1617--1622, 1988.
\newblock \href {https://doi.org/10.1109/49.12889}
  {\path{doi:10.1109/49.12889}}.

\bibitem{wolsey2020}
Laurence~A Wolsey.
\newblock {\em Integer Programming}.
\newblock John Wiley \& Sons, 2020.

\bibitem{zhu2012}
Xiaoyan Zhu and Wilbert~E Wilhelm.
\newblock A three-stage approach for the resource-constrained shortest path as
  a sub-problem in column generation.
\newblock {\em Computers \& Operations Research}, 39(2):164--178, 2012.
\newblock \href {https://doi.org/10.1016/j.cor.2011.03.008}
  {\path{doi:10.1016/j.cor.2011.03.008}}.

\end{thebibliography}

\newpage

\appendix
\renewcommand\thefigure{A\arabic{figure}}    
\setcounter{figure}{0} 
\renewcommand{\thetable}{A\arabic{table}}
\setcounter{table}{0}
\section{APPENDIX}
\FloatBarrier

\centering
        \captionof{table}{Top: Share (\%) of optimally solved \instance{ER} graphs. Bold marks the better of \algo[Top]{PB} and \algo{\ABplus}. ``--'' indicate these instances were only investigated during pilot studies.}
        \label{appendix:H2:table}
        \smaller
    \setlength\tabcolsep{4.3pt}
    \begin{tabular}{|l*{4}{|ccc}*{2}{|cc}|} %
     \tableHead
          \hline
          \hline

        \rowcolor{lightGray}
        \algo[Top]{PB} &  $\bm{100} $ & $ \bm{100} $ & $ \bm{97} $ & $ \bm{100} $ & $ \bm{100} $ & $ \bm{93} $ & $ \bm{100} $ & $ \bm{100} $ & $ \bm{93} $ & $ \bm{70} $ & $ \bm{59} $ & $ \bm{39} $ & $ 59 $ & $ 15 $ & $ 18 $ & $ 9$\\
          
        \algo[Orig]{PB}  & $49 $ & $ 34 $ & $ 19 $ & $ 43 $ & $ 29 $ & $ 15 $ & $ 44 $ & $ 34 $ & $ 16 $ & $ 39 $ & $ 24 $ & $ 12 $ & $ 21 $ & $ 13 $ & $ 12 $ & $ 5$  \\

        \rowcolor{lightGray}
      \algo{\ABplus} &  $93 $ & $ 93 $ & $ 86 $ & $ 84 $ & $ 80 $ & $ 70 $ & $ \bm{100} $ & $ 93 $ & $ 90 $ & $ 57 $ & $ 51 $ & $ 33 $ & $ \bm{100} $ & $ \bm{71} $ & $ \bm{48} $ & $ \bm{14}$\\
          
    \hline
    
    \end{tabular}

\vspace{2em}

\centering
   \input{figures/RD_W_t_numVars.pgf}
 \captionof{figure}{Number of variables of \algo[Top]{PB} and \algo{\ABplus} on \instance{Road} instances with \weight{t} (i.e., travel times).}
   \label{fig:RD_W_t_numVars}

    \centering
        \captionof{table}{Corresponding to \Cref{table:shareSolvedWM}. Top and middle repeat information. Top: Share (\%) of optimally solved \instance{WM} graphs. Bold font marks the better of \algo[Top]{PB} and \algo{\ABplus}. Middle: median gap of \algo{BG} to the optimum solution value (or the best found dual bound; cf.\ top). 
        Bottom: interquartile range of the gap. A ``--'' indicates that there are instances in that class with no non-trivial lower bound.}
        \label{appendix:H4:table1bottom}
        \smaller
    \setlength\tabcolsep{4.3pt}
    \begin{tabular}{|l*{4}{|ccc}*{2}{|cc}|} %
     \tableHead
          \hline
          \hline

        \rowcolor{lightGray}
        \algo[Top]{PB} &  $\bm{100}$ & $\bm{100}$ &  $\bm{99}$ & $\bm{100}$ &  $\bm{100}$ & $\bm{90}$ &  $\bm{100}$ & $\bm{100}$ & $\bm{75}$ &  $\bm{62}$ &  $\bm{49}$ & $25$ & $41$ & $10$ & $19$ & $8$\\
          
      \algo{\ABplus} &  $93$ & $93$ &  $86$ & $86$ &  $81$ & $66$ &  $93$ & $93$ & $74$ &  $57$ &  $46$ & $\bm{30}$ & $\bm{100}$ & $\bm{62}$ & $\bm{48}$ & $\bm{13}$\\

    \hline
    \hline
    $n$ & \multicolumn{16}{c|}{Median gap of \algo{BG} to best lower bound in \% }\\
    \hline
    \hline
    \rowcolor{lightGray}
    $100$ & $0.0$ & $2.1$ & $5.0$ & $0.0$ & $2.3$ & $4.6$ & $0.5$ & $4.0$ & $6.1$ & $2.8$ & $8.6$ & $8.0$ & $5.4$ & $31.1$ & $19.8$ & $129.1$\\

    $500$ & $0.0$ & $1.3$ & $3.5$ & $0.1$ & $2.0$ & $4.9$ & $0.1$ & $3.2$ & $5.1$ & $1.4$ & $6.3$ & -- & $0.6$ & $7.3$ & $4.7$ & $90.2$\\

    \rowcolor{lightGray}
    $2000$ & $0.0$ & $0.6$ & $3.2$ & $0.0$ & $1.1$ & $3.6$ & $0.0$ & $1.1$ & $2.3$ & $0.7$ & -- & -- & $0.1$ & $1.3$ & $3.6$ & $56.4$ \\
    \hline
    \hline

    $n$ & \multicolumn{16}{c|}{Interquartile range of the gap of \algo{BG} to best lower bound in \% }\\
    \hline
    \hline
    \rowcolor{lightGray}
    $100$ & $ 0.0 $ & $ 2.8 $ & $ 4.2 $ & $ 0.6 $ & $ 3.1 $ & $ 4.4 $ & $ 1.4 $ & $ 4.6 $ & $ 3.4 $ & $ 2.8 $ & $ 3.5 $ & $ 2.3 $ & $ 10.1 $ & $ 65.2 $ & $ 4.9 $ & $ 24.6$ \\

    $500$ & $0.1 $ & $ 1.1 $ & $ 1.1 $ & $ 0.1 $ & $ 1.2 $ & $ 1.9 $ & $ 0.6 $ & $ 2.8 $ & $ 3.7 $ & $ 0.3 $ & $ 1.8 $ & -- & $ 0.9 $ & $ 17.6 $ & $ 0.5 $ & $ 13.3$\\

    \rowcolor{lightGray}
    $2000$ & $0.0 $ & $ 0.9 $ & $ 1.0 $ & $ 0.0 $ & $ 0.4 $ & $ 0.6 $ & $ 0.1 $ & $ 1.4 $ & $ 0.2 $ & $ 0.2 $ & -- & -- & $ 0.1 $ & $ 2.6 $ & $ 0.2 $ & $ 2.7$ \\
    \hline

    \end{tabular}

\vspace{2em}
    \centering
        \captionof{table}{Median gap of \algo{BG} to the optimum solution value (or the best found dual bound).}
        \label{appendix:H4:tableSpecial}
        \smaller
    \setlength\tabcolsep{3.4pt}
    \begin{tabular}{|c|c|c|c|c|c|c|c|c|c|} 
      \hline
                 & \multicolumn{7}{c|}{\instance{SteinLib}} & \multicolumn{2}{c|}{\instance{Road}} \\
         \hline
        $\alpha$ & Cross-Grid                                & Group ST                              & Rectilinear & VLSI   & Sparse inc & Sparse rand & other & ~~~~\weight{t}~~~~ & ~~~\weight{len}~~~ \\
        \hline\hline
          \rowcolor{lightGray}
        $1.2$    & $0.0$                                     & $0.0$                                 & $0.0$       & $0.0$  & $0.0$       & $0.0$         & $0.0$ & $0.03$     & $0.01$       \\
        $1.5$    & $0.0$                                     & $0.0$                                 & $0.0$       & $0.0$  & $0.0$       & $0.0$         & $0.6$ & $0.11$     & $0.06$       \\
          \rowcolor{lightGray}
        $2$      & $84.8$                                    & $0.0$                                 & $0.9$       & $15.7$ & $0.3$       & $0.5$         & $4.6$ & $0.50$     & $0.51$       \\
     \hline
    \end{tabular}

\end{document}